\documentclass[sigconf, screen]{acmart}

\acmSubmissionID{407}

\usepackage{hyperref}
\usepackage{latexsym}
\usepackage{mathrsfs}
\usepackage{multirow}
\usepackage{amsmath}
\usepackage{amsfonts}
\usepackage{subfigure}
\usepackage{graphicx}
\usepackage{tabularx}
\usepackage[boxed,ruled,lined]{algorithm2e}
\usepackage{algorithmic}
\usepackage{comment}
\usepackage{balance}
\usepackage{bm}
\usepackage[inline]{enumitem}
\usepackage[skip=1mm]{caption}

\usepackage{acronym}
\usepackage{subcaption}
\newtheorem{theorem}{Theorem}

\newcommand{\ie}{\emph{i.e., }}

\newcommand{\wrt}{\emph{w.r.t. }}

\acrodef{IR}{information retrieval}

\usepackage{amsmath,amsfonts,bm}

\def\eqref#1{equation~\ref{#1}}
\def\1{\bm{1}}

\DeclareMathAlphabet{\mathsfit}{\encodingdefault}{\sfdefault}{m}{sl}
\SetMathAlphabet{\mathsfit}{bold}{\encodingdefault}{\sfdefault}{bx}{n}

\DeclareMathOperator*{\argmax}{arg\,max}

\DeclareMathOperator{\sign}{sign}

\author{Chen Xu}
\orcid{https://orcid.org/0000-0002-3070-9358}
\affiliation{%
  \institution{Gaoling School of Artificial Intelligence\\Renmin University of China}
  \city{Beijing}
  \country{China}
}
\email{xc\_chen@ruc.edu.cn}

\author{Wei Chu}
\orcid{https://orcid.org/0009-0004-4534-8075}
\affiliation{%
  \institution{\mbox{Gaoling School of Artificial Intelligence}\\Renmin University of China}
  \city{Beijing}
  \country{China}
}
\email{wchu4112@gmail.com}

\author{Wenyu Hu}
\orcid{https://orcid.org/0009-0005-7130-9894}
\affiliation{%
  \institution{\mbox{School of Mathematics}\\Renmin University of China}
  \city{Beijing}
  \country{China}
}
\email{huwenyu\_123@ruc.edu.cn}

\author{Fengran Mo}
\orcid{0000-0002-0838-6994}
\affiliation{%
  \institution{\mbox{University of Montreal}}
  \city{Quebec}
  \country{Canada}
}
\email{fengran.mo@umontreal.ca}

\author{Jun Xu}
\orcid{https://orcid.org/0000-0001-7170-111X}
\authornote{Jun Xu is the corresponding author. }
\affiliation{%
    \institution{\mbox{Gaoling School of Artificial Intelligence}\\Renmin University of China}
  \city{Beijing}
  \country{China}
}
\email{junxu@ruc.edu.cn}
\author{Maarten de Rijke}
\orcid{https://orcid.org/0000-0002-1086-0202}
\affiliation{
  \institution{University of Amsterdam}
  \city{Amsterdam}
  \country{The Netherlands}
}  
\email{m.derijke@uva.nl}

\renewcommand{\shortauthors}{Chen Xu et al.}

\copyrightyear{2026}
\acmYear{2026}
\setcopyright{cc}
\setcctype{by}
\acmConference[SIGIR '26]{Proceedings of the 49th International ACM SIGIR Conference on Research and Development in Information Retrieval}{July 20--24, 2026}{Melbourne, VIC, Australia}
\acmBooktitle{Proceedings of the 49th International ACM SIGIR Conference on Research and Development in Information Retrieval (SIGIR '26), July 20--24, 2026, Melbourne, VIC, Australia}
\acmDOI{10.1145/3805712.3809614}
\acmISBN{979-8-4007-2599-9/2026/07}



\title{Equivalence of Fair Re-ranking and Walrasian Equilibrium: A Manifold Perspective}
\title{The Attention Market: Interpreting Online Fair Re-ranking as Manifold Optimization under Walrasian Equilibrium}

\begin{document}

\begin{abstract}
Fair re-ranking aims to promote long-tail items and enhance diversity within groups in information retrieval. 
While previous research on online fairness-aware re-ranking has shown promising outcomes, our comprehensive evaluation of online fair re-ranking methods over 20 settings reveals significant performance disparities among existing methods. To uncover the root causes of these inconsistencies, we reformulate fair re-ranking within an attentional market framework governed by a Walrasian Equilibrium, where the fairness is treated as a taxation cost. This market-based formulation is then coupled with manifold optimization, demonstrating that seeking this equilibrium is equivalent to performing gradient descent on a specific ranking manifold constructed by the market. 
Different re-ranking settings induce distinct manifold geometries, and these intrinsic geometric differences dictate the gradient landscapes and optimization trajectories. 
We propose ManifoldRank, an efficient online fair re-ranking algorithm. ManifoldRank adjusts gradients to align with the ranking manifold, considering various contextual settings. On the supply side, it incorporates a gradient adjustment based on different fairness requirements, accounting for associated costs. On the demand side, it empirically predicts an additional gradient adjustment term derived from the ranking scores. By integrating these two gradient adjustments, ManifoldRank effectively balances fairness and accuracy. Experimental results across multiple datasets confirm ManifoldRank's effectiveness.
\end{abstract}

\maketitle

\section{Introduction}\label{sec:intro}

Fair re-ranking plays a crucial role in promoting long-tail items and enhancing diversity within groups in information retrieval (IR)~\cite{lifairness, fairrec}. Generally, re-ranking algorithms operate under either offline or online paradigms. Offline solutions typically optimize ranking slots based on global user information, thereby achieving better performances~\cite{TaxRank, TaoSIGIRAP}. However, real-world industrial systems operate under strict latency constraints, which necessitate the use of online solutions. To achieve this, prior work often uses online gradient descent to dynamically adjust ranking scores based on historical interaction data~\cite{xu2023p, ElasticRank, cpfair}.

While previous work has demonstrated promising results, evaluations are often restricted to limited settings, such as a narrow range of datasets or base models. To investigate the stability and robustness of online fair re-ranking approaches, we conduct a comprehensive evaluation spanning over 20 re-ranking settings, including the seven most widely used datasets and 3 distinct base ranking models. Within each setting, we tuned over 50 times across a sweep of accuracy-fairness trade-off parameters. Detailed experimental settings are provided in Section~\ref{sec:exp_settings}. Figure~\ref{fig:intro} details the evaluation results, with models plotted on the x-axis and performance ranking on the y-axis (where lower values indicate better performance). The box plots depict the performance distribution, including the median, quartiles, and range. These results reveal a significant lack of consistency in algorithmic performance and the sensitivity of existing online approaches to fair re-ranking to diverse environments.

\begin{figure}[h]
    \centering    
    \includegraphics[width=\linewidth]{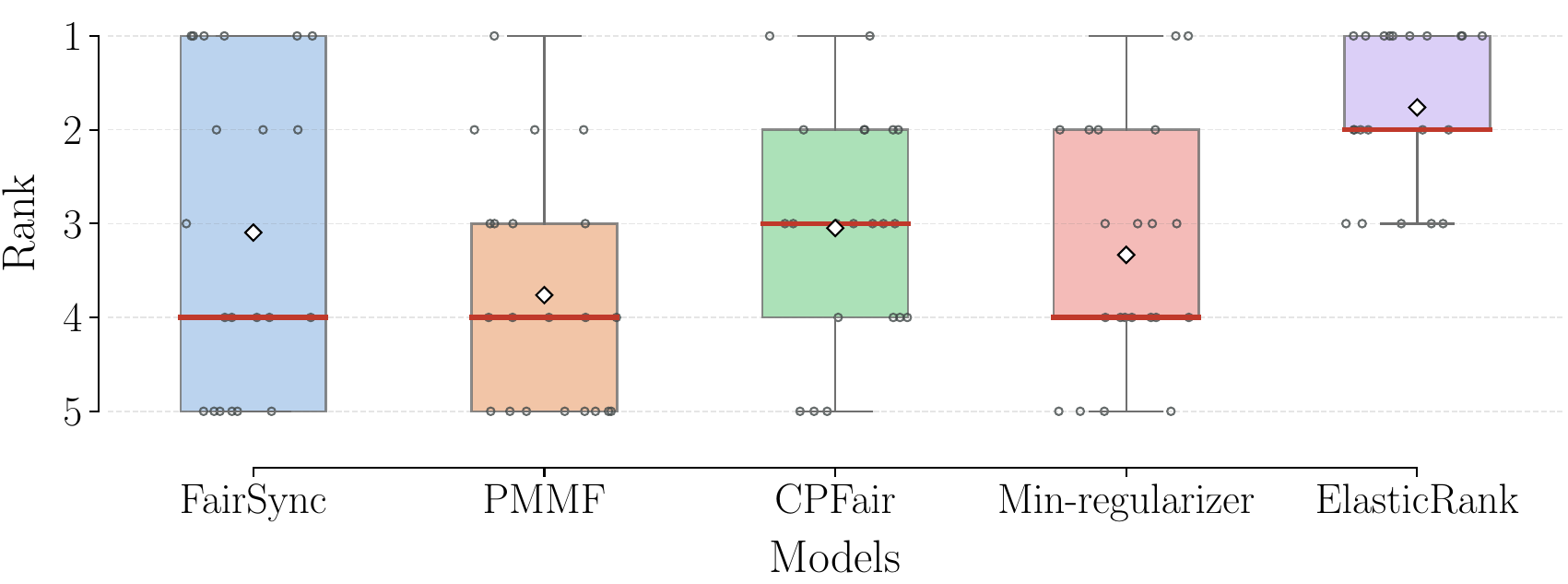}
    \vspace*{-1mm}
    \caption{Performance rankings (lower is better) of fair online re-ranking models under multiple re-ranking settings, including seven datasets and three base models.}
    \label{fig:intro}
    %\vspace{-0.5cm}
\end{figure}

\smallskip\noindent%
\textbf{Theoretical analysis.} Drawing from economic principles, we will apply the concept of attention economics to examine this issue.
Attention economics is a theoretical framework that treats human attention as a scarce and valuable commodity that must be efficiently allocated~\citep{simon-1971-designing}.
Prior work~\cite{TaxRank, ElasticRank} has adopted this framework to model fairness as a form of economic taxation.
We expand on this line of work and draw a parallel between the traditional economic market and the fair re-ranking market, see Table~\ref{tab:market_comparison}.
In our formulation, user satisfaction functions as the market demand, where users seek to maximize the utility of their limited attention. Conversely, content provider fairness is modeled as a taxation mechanism, introducing costs for producing attention that the user needs.
The optimal re-ranking score acts as the equilibrium price balancing these forces, a so-called Walrasian Equilibrium~\cite{arrow2024existence}. Crucially, we establish that seeking this equilibrium is equivalent to gradient descent on a specific re-ranking manifold. 
Consequently, performance inconsistency is explained by the fact that different re-ranking settings induce distinct manifold geometries and optimization trajectories. 
That is, we attribute the root cause of the performance inconsistency to the neglect of the underlying geometric properties of the optimization landscape in economics. 
See Section~\ref{sec:theory} for a detailed theoretical analysis of such a comparison.

\begin{table}[t]
    \centering
    \caption{Analogy between the economic market and the fair re-ranking market (attention economy).}
    \label{tab:market_comparison}
    \begin{tabularx}{\linewidth}{@{} XX @{}}
        \toprule
        \textbf{Economic market} & \textbf{Fair re-ranking market} \\
        \midrule
        Buying products (\textit{demand}) & Consuming attention (\textit{users}) \\
        Producing products (\textit{supply}) & Producing attention (\textit{providers}) \\
        Binding income & Limited attention \\
        Taxation mechanism & Fairness function \\
        Selling price & Re-ranking score \\
        Seeking equilibrium & Optimizing re-ranking  \\
        Market context & Manifold geometries \\
        \bottomrule
    \end{tabularx}
    \vspace*{-2mm}
\end{table}

%Here is the motivation in Figure~\ref{fig:demand_supply}: Sub-figure (a) shows the typical economic Equilibrium: demand side will less likely to decrease the consuming volumn ($Q$) if the price $P$ is higher, but the supplier will increase the production volume. The optimal price is the Equilibrium point of two plots.

%Compared to the traditional economic market, the fair re-ranking market (shown in Sub-figure (b)), user satisfaction serves as the demand function for attention, while content providers represent the supply side, competing for that attention subject to varying taxation costs determined by fairness constraints. The optimal re-ranking score (price) is also the Equilibrium point of two plots. 
%Crucially, we establish a link to manifold optimization: seeking this market equilibrium is isomorphic to gradient descent on a specific ranking manifold. Thus, the inconsistency is explained by the fact that different re-ranking settings define distinct manifold geometries. Detailed theoretical analysis is in Section XXX.

\smallskip\noindent%
\textbf{Empirical analysis.} To further investigate how different settings influence manifold geometries and optimization trajectories, we conduct an ordinal regression analysis to identify which factors of the ranking scores significantly contribute to algorithm performance inconsistency. Given the difficulty of directly modeling the manifold induced by ranking scores, we instead used statistical features of the ranking score matrix as proxy variables. Our results indicate a strong correlation between the entropy and skew of the scores and the gradient step size. Consequently, we propose to calibrate the gradient magnitude in our online fair re-ranking algorithm by jointly considering the fairness (supply side) and statistical features of ranking scores (demand side). 
%A detailed analysis can be found in Section \mdr{XXX}.

Building on this observation, we propose ManifoldRank, an efficient online fair re-ranking algorithm that aligns optimization gradients with the ranking manifold by incorporating two adjustment terms. The algorithm first introduces a fairness-based gradient term on the supply side to identify the equilibrium direction. Second, it empirically estimates a gradient step adjustment derived from the statistical features of the scores to fit the underlying manifold structure from the demand side. By combining these two terms, ManifoldRank achieves both efficient and effective fair re-ranking. Extensive experiments verify the effectiveness of the ManifldRank method across diverse settings. Our source code and all baselines are integrated into the toolkit FairDiverse~\cite{xu2025fairdiverse} in GitHub~\url{https://github.com/XuChen0427/FairDiverse}.

\smallskip\noindent%
\textbf{Main contributions.}
We summarize our major contributions:
\begin{enumerate}[leftmargin=*]
    \item We theoretically frame fair re-ranking as a Walrasian Equilibrium within an attention economy. We demonstrate that the process of seeking this market equilibrium can be viewed as a process analogous to gradient descent on a ranking manifold.

    \item We propose an online fair re-ranking algorithm, ManifoldRank, which aligns optimization gradients with the manifold structure by incorporating two adaptive adjustment fairness terms.

    \item Extensive empirical evaluations on seven publicly available ranking datasets demonstrate that ManifoldRank consistently surpasses state-of-the-art baselines.
\end{enumerate}

%To analyze the root cause for such an inconsistency, we find out that this problem can be explained as an optimization process under the economic theory. Inspired by previous work~\cite{TaxRank, ElasticRank}, which regards the fairness as a taxation form in economics, we firstly mathematically show that fair re-ranking tasks can be equivelently re-written as a seeking Walrasian Equilibrium process in attention market of economics. Specifically,  user satisfaction can be viewed as the demand function that consuming the their attention and supply side providers want to gain user's attention. Fairness degree is servered as a taxation cost on the different user attention to support the tail providers gain more attention from users. The economic formulation is then bridged with manifold optimization, demonstrating that seeking this equilibrium is isomorphic to performing gradient descent on a specific ranking manifold.  And varying re-ranking settings induce distinct manifold geometries, which results in the in-consistency for different set
\section{Related Work}

\noindent%
\textbf{Fair re-ranking.}
Fairness is a central concern in information retrieval (IR), prompting extensive research into embedding fairness constraints within ranking systems~\cite{wang2023survey, lifairness, fang2024fairness, wu2023fairness, GOMEZ2025100311}. Existing literature generally categorizes fair ranking algorithms into three pipeline stages. First, \textit{pre-processing} methods, such as causal models~\cite{yang2020causal} and probabilistic clustering~\cite{lahoti2019ifair}, attempt to remove bias from input data prior to training~\cite{rus-2024-annorank}. Second,  \textit{in-processing} methods integrate constraints directly into the training process via techniques like re-weighting~\cite{APR, jiang2024item} and regularization~\cite{yao2017beyond}. Third,  \textit{post-processing} methods (\ie re-ranking) adjust the final candidate list after the model has been trained~\cite{xu2023p, fairrec, wu2021tfrom, liu2025repeat}. 
In this paper, we focus on fairness in the post-processing stage, as re-ranking offers a flexible and effective mechanism for ensuring fairness without re-training the underlying base ranking models.

\smallskip\noindent%
\textbf{Online fair re-ranking.} Re-ranking strategies can be split into offline and online paradigms. While offline methods~\cite{liu2025repeat, fairrec, fairrecplus, nips21welf, TaxRank} use mixed-integer linear programming for global optimization, they are computationally expensive and impractical for real-time needs.  Consequently, online algorithms are essential for scenarios requiring immediate user satisfaction~\cite{xu2023p}. MACE proposes causal embedding learning for fair streaming IR systems to improve OOD generalization \cite{zhang2024model}. 
These include heuristic methods like \textit{CP-Fair}~\cite{cpfair}, which use greedy solutions to obtain the fair scores. On the other hand, gradient-based methods~\cite{ElasticRank, fairsync, balseiro2021regularized, xu2023p} that use online optimization. For example, \textit{P-MMF}~\cite{xu2023p} employ online learning techniques to optimize fairness objectives via sub-gradients, and \textit{ElasticRank}~\cite{ElasticRank} uses the elasticity of ranking as the online gradient.
However, current online approaches tend to be ranking-score-free; they fail to account for the geometry of the optimization manifold constructed by the ranking scores.

\smallskip\noindent%
\textbf{Multi-stakeholder fairness.} 
Depending on the stakeholders involved~\cite{abdollahpouri2020multistakeholder, abdollahpouri2019multi}, re-ranking fairness can be categorized into user-side or provider-side. User fairness typically focuses on invariance to sensitive user attributes~\cite{rus2023counterfactual, li2021user, wang2021user}. Item (or provider) fairness, however, aims to make item exposure more equitable~\cite{xu2023p, ElasticRank}. We prioritize item fairness in this work because it presents a unique challenge in online settings: unlike user fairness, item fairness requires managing cumulative exposure that accrues dynamically over a period of time, such as one day.

\smallskip\noindent%
\textbf{Economics and fair re-ranking.} 
Fairness has long been a central subject of study in economics, particularly through the lenses of taxation and social welfare functions~\cite{ng1983welfare}. Economic study the fairness mainly through the demand-supply equilibrium analysis method~\cite{atkinson1974pigou}.
In the context of ranking, prior work~\cite{saito2022fair} treats fair ranking as a resource allocation problem, formulating it via the economic concept of Nash Social Welfare~\cite{price4fairness}. Similarly, other studies~\cite{TaxRank, ElasticRank} model fair re-ranking as a taxation process, a key mechanism for redistribution. However, these approaches typically adopt economic concepts only at a conceptual level. In contrast, in this paper we propose a rigorous mathematical comparison establishing the formal equivalence between economic theory and fair re-ranking formulations.

\section{Preliminaries}

In this section we define the fair re-ranking tasks and introduce our framework for modeling economic markets.

\subsection{Fair re-ranking}\label{sec:formulation}

In re-ranking tasks, let $\mathcal{U}$ denote the set of users and $\mathcal{I}$ denote the set of items. Each item $i \in \mathcal{I}$ belongs to a unique group $g \in \mathcal{G}$, and the set of items within a specific group $g$ is represented as $\mathcal{I}_g$. When a user $u \in \mathcal{U}$ accesses the system, a candidate list $L_K(u)$ is generated with ranking scores $s_{u,i}$ for each item $i \in L_K(u)$, $K$ is the size of the candidate list. The ranking score $s_{u,i}$ represents the relevance of item $i$ to user $u$, is typically pre-computed based on ranking models, such as \textit{DMF}~\cite{DMF}, \textit{SASRec}~\cite{SASRec}.

To formulate the re-ranking decision, we introduce a binary decision variable $x_{u,i} \in \{0, 1\}$. Here, $x_{u,i}=1$ indicates that item $i$ is displayed to user $u$ (i.e., selected in the final ranked list), and $x_{u,i}=0$ otherwise. Consequently, the user and item group utilities are determined by the displayed items. The item group utility $\bm{v}_g$ and user utility $\bm{w}_u$ are defined as the accumulated scores of the displayed items:
\begin{equation}\label{eq:utility_define}
    \bm{v}_g = \sum_{u\in\mathcal{U}}\sum_{i\in L(u)} x_{u,i} \cdot s_{u,i} \cdot \mathbb{I}(i\in \mathcal{I}_g), \quad \bm{w}_u = \sum_{i \in L(u)} x_{u,i} \cdot s_{u,i},
\end{equation}
where $\mathbb{I}(\cdot)$ is the indicator function.

The objective of the fair re-ranking problem is to maximize the overall user utility (\ie social welfare) while ensuring that the exposure or utility distributed among different item groups satisfies fairness constraints. This problem can be formulated as a constrained optimization problem:

\begin{equation}
\begin{aligned}
    \max_{x_{u,i}\in \mathcal{X}}   \sum_{u \in \mathcal{U}} \bm{w}_u, \textrm{ such that } \bm{v}_r \approx \bm{v}_p, \quad \forall r, p \in \mathcal{G},
\end{aligned}
\end{equation}
where $\mathcal{X} = \{x_{u,i}\mid x_{u,i} = K,~ x_{u,i} \in \{0, 1\}\}$ represents the set of binary decision variables. The first constraint enforces that the utility difference between any pair of groups $r$ and $p$ remains within an acceptable range. To enforce the constraint $\bm{v}_r \approx \bm{v}_p$, a standard approach involves maximizing a fairness function $r'(v)$, which is formulated such that its optimum yields $\bm{v}_r \approx \bm{v}_p$.

% \emph{Offline settings.}
Typical online fair re-ranking methods aim to learn a fair re-ranking score $f_{u,g}$ with the ranking score to achieve an objective like the following:
\begin{equation}
    L_K(u) = \argmax_{S \in \mathcal{I}^K} \sum_{i \in S} \left(s_{u,i}-\hat{f}_{u, g}\right),
\end{equation}
where $S$ is the subset of $\mathcal{I}$ with size $K$ and $g$ indicates that item $i$ belongs to the group $g$: $i\in\mathcal{I}_g$.

\subsection{Economic markets}
\label{sec:econ_market}

In microeconomics, market forces are driven by the interaction between buyers (demand) and sellers (supply). Figure~\ref{fig:demand_supply}(a) provides an intuitive illustration of this concept.

\emph{1. The demand curve (The buyer)} represents how much consumers want to buy. It follows the \textit{Law of Demand}: when prices are lower, people buy more.

\paragraph{Example:} Imagine buying apples. If an apple costs \$5, you might buy none. If it drops to \$0.50, you might buy 10 to make a pie.
Mathematically, the quantity demanded ($Q_d$) is a decreasing function of price ($P$). The mathematical formulation is:
\[
Q_d = f_d(P), \quad \frac{d Q_d}{dp} \leq 0.
\]

\emph{2. The supply curve (The seller)} represents how much producers are willing to sell. It follows the \textit{Law of Supply}: higher prices make production more profitable, so producers will supply more product.

\paragraph{Example:} Imagine you are an apple farmer. If apples sell for \$0.10, it is not worth the effort, so you supply 0. If they sell for \$5, you work harder to sell as many as possible.
Mathematically, the quantity supplied ($Q_s$) is an increasing function of price:
\[
Q_s = f_s(P), \quad \quad \frac{d Q_s}{dp} \ge 0.
\]

\emph{3. Market Equilibrium} occurs when the buyer's desire matches the seller's willingness. At this optimal price, there is no shortage and no surplus.

\paragraph{Example:} At a price of \$2, the farmer brings exactly 100 apples, and customers want exactly 100 apples. If the price rises above \$2, a surplus of unsold apples forces the farmer to lower the price to clear his stock. Conversely, if the price falls below \$2, a shortage creates competition among buyers, naturally driving the price back up until it returns to the equilibrium level.
Mathematically, we set demand equal to supply: $Q_d (P^*) = Q_s (P^*)$, and solve for the equilibrium point: $E = (Q^*, P^*)$.

\begin{figure}[t]  
    \centering    
    \includegraphics[width=\linewidth]{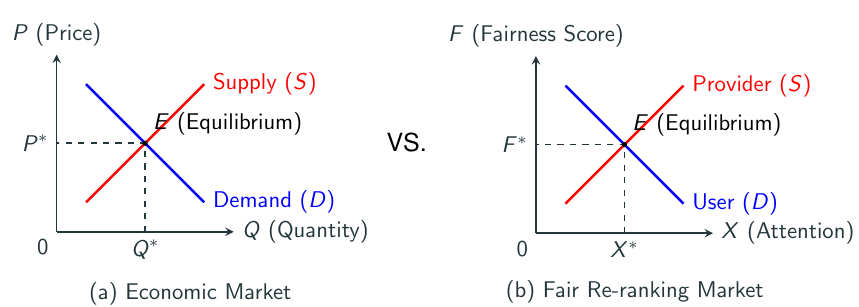}
    \caption{Parallels between (a) an economic market, where the x-axis represents production/consumption quantity and the y-axis represents selling price, and (b) a fair re-ranking market, where the x-axis represents the user attention for an item and the y-axis represents the fair re-ranking score.}
    \label{fig:demand_supply}
\end{figure}

Achieving market equilibrium represents a state of allocative efficiency where total economic surplus is maximized, corresponding to a Pareto optimal allocation of resources as described by the First Fundamental Theorem of Welfare Economics~\cite{ng1983welfare}. In the next section we use these basic notions to formalize the online fair re-ranking problem.
\section{Theory: Attention Market for Fair Re-ranking}
\label{sec:theory}

In this section, we reformulate the fair re-ranking task as a market for selling and consuming users' attention, while explaining why this market-based perspective is crucial for optimizing fair re-ranking. We model the fair re-ranking market by comparing it to the economic market discussed in Section~\ref{sec:econ_market}.

\emph{Overall idea.} The key distinction between a traditional economic market and fair re-ranking lies in redefining exchange variables. The price is interpreted as the fairness score $f_g$, regulating the trade-off between accuracy and fairness, while the quantity corresponds to user attention $x_{u,i}$, treated as a scarce resource allocated to items (see Figure~\ref{fig:demand_supply}(b)). We next elaborate on this formulation.

%\subsection{Equilibrium in attention market}
% \subsection{Overall idea}
% The fundamental distinction between a traditional economic market and the fair re-ranking market lies in the redefinition of exchange variables. The currency price acts as the fair re-ranking score $f_g$, which acts as a regulatory mechanism to balance user accuracy with provider fairness. Furthermore, the Quantity $Q$ shifts from the volume of products sold to the user's attention volume $x_{u,i}$. In this context, attention functions as a scarce resource that users ``spend'' or allocate to content providers. An intuitive visualization can be seen in Figure~\ref{fig:demand_supply} (b). Next, we explain the overall idea in more detail.

\subsection{Demand side}\label{sec:demand}
For each user $u$, we can rewrite their goal to maximize their individual utility for paying attention cost (i.e., global fairness score) $f_{g}$ for group $g\in\mathcal{G}$:
\begin{equation}\label{eq:demand}
      \begin{aligned}
        X^d = \argmax_{x_{u,i}\in\mathcal{X}} &\sum_{i\in\mathcal{I}} \underbrace{\bm{s}_{u,i}\vphantom{f_g}}_{\text{ranking score}}x_{u,i}-\textcolor{blue}{\underbrace{f_g}_{\text{price of paying attention}}}x_{u,i}.
    \end{aligned}
\end{equation}

The user $u$ operates under a constrained attention budget (e.g., the capacity to view only $K$ items) and allocates this scarce resource to maximize their cumulative utility derived from $\bm{s}_{u,i}$. Within this framework, the fairness score $f_g$ functions as the market price, representing the cost incurred when allocating attention $x_{u,i}$ to a specific item.

\begin{theorem}[Law of Demand]\label{theo:law_of_demand}
    Let $X_g^d$ denote the aggregate attention consumed by users. We have the following law of demand. User attention consumption is inversely related to the fairness score $f_g$, implying:
\[
    X_g^d = \sum_{u} \sum_{i \in X^d} x_{u,i} \mathbb{I}(i \in \mathcal{I}_g),\quad \frac{\partial X_g^d}{\partial f_g} \le 0.
\]    
\end{theorem}

\noindent%
A detailed proof is provided in Appendix~\ref{app:law_of_demand}. Intuitively, a higher price $f_g$ reduces the attractiveness of item $i$, leading to less or equal attention allocation $x_{u,i}$, which in turn decreases or maintains the consuming attention $X_g^d$ for provider $g$.

\subsection{Supply side}\label{sec:supplyside}

For the supply side, each provider $g$ produces attentions that user needs based on the price of getting attention (\ie fairness score):

\begin{equation}\label{eq:supply}
    \begin{aligned}
     X^s &= \argmax_{x_{u,i}\in\mathcal{X}} \sum_{u}\sum_{i\in\mathcal{I}_g} \textcolor{blue}{\underbrace{f_g}_{\text{price of getting attention}}}x_{u,i} - \textcolor[HTML]{006400}{\underbrace{r(v_g)}_{\text{taxation cost}}}.
     %& \textrm{s.t. }\quad \sum_i x_{u,i}\leq K, ~x_{u,i}\in [0,1].
\end{aligned}    
\end{equation}
In this formulation, $r(v_g)$ represents the fairness penalty function for different group utility $v_g$ (defined in Eq.~(\ref{eq:utility_define}), which relates to the attention $x_{u,i}$). Please note that $f_g$ here should be consistent with Equation~(\ref{eq:demand}), as the total revenue from sales must equal the total amount consumed.

Conceptually, this parallels the mechanism of taxation in economics~\cite{ElasticRank}, where a cost is imposed to discourage inequitable resource allocation. Next, we define the $r(v_g)$ as:
 \begin{equation}\label{eq:taxation_cost}
        r(v_g) = \sign(\alpha\beta)v_g^{\textcolor{blue}{\beta}}\left(\sum_g v_g^{\textcolor{blue}{\beta}}\right)^{\textcolor{red}{\alpha-1}},
\end{equation}
where $\alpha \in \mathbb{R}$ is the global taxation and $\beta \in \mathbb{R}$ is the local taxation rate. 
This definition will ensure the previous fairness function is equivalent to the taxation cost of the market, according to the following theorem:

\begin{theorem}[Fairness function as taxation]\label{theo:taxation}
    The total market cost $r'(v) = \sum_g r(v_g)$ satisfies the taxation requirement~\cite{atkinson1974pigou}:
    %\emph{Positive Marginal Tax Rate:} 
    The first derivative represents the marginal tax rate. This condition implies that the tax liability is strictly increasing with income:
    $
        \frac{\partial r'(v)}{\partial v_g} \ge 0.
    $
\end{theorem}

 % \begin{itemize}[leftmargin=*, align=left]
    %     \item \emph{Positive Marginal Tax Rate:} The first derivative represents the Marginal Tax Rate (MTR). This condition implies that the tax liability is strictly increasing with income:
    %     \[
    %         \frac{\partial r'(v)}{\partial v_g} \ge 0,
    %     \]
    %     \item \emph{Increasing Marginal Tax Rate}: The second derivative represents the rate of change of the marginal tax rate. This condition implies that the tax function is Convex:
    %      \[
    %         \frac{\partial^2 r'(v)}{\partial v_g^2} \ge 0.
    %      \]
    % \end{itemize}

    % $\mathbf{x} = t \mathbf{v}, t \geq 0, f(t \mathbf{v}) = c t^{a b}(c > 0)$

\noindent%
Once we have established this supply side, we can effectively generalize the market cost $r'(v)$ to encompass any fairness metric listed in Table~\ref{tab:general_fairness}.
And we can easily obtain the Law of Supply:

\begin{table}[t]
    \centering
    \caption{Generalization of fairness metrics via $\alpha$ and $\beta$}
    \label{tab:general_fairness}
    \begin{tabular}{l c c}
        \toprule
        \textbf{Fairness metric $r'(v)$} & $\alpha$ & $\beta$ \\
        \midrule
        \emph{Max-min fairness}~\citep{xu2023p} $\min_g(v_g)$ & $\infty$ & $-\infty$ \\
        
        \emph{$p$-norm}~\citep{bektacs2020using} $\left(\sum_{i=1}^g v_i^p\right)^{1/p}$ & $1/p$ & $p$ \\
        
        \emph{Elastic fairness}~\citep{ElasticRank} $\text{sign}(1-t)\left(\sum_{i=1}^g\bar{\bm{v}}_i^{1-t}\right)^{(1/t)}$ & $1/t$ & $1-t$ \\
        
        \emph{$\alpha$-fairness}~\citep{TaxRank} $\sum_i v_i^{1-\alpha}$ & $1$ & $1-\alpha$ \\
        
        \emph{Proportional fairness}~\citep{price4fairness} $\sum_i \log(v_i)$ & $1$ & $\beta \to 0$ \\ 
        \bottomrule
    \end{tabular}
\end{table}

\begin{theorem}[Law of Supply]\label{theo:law_of_supply}
    Let $X_g^s$ be the aggregate attention allocated to provider $g$. We posit a law of supply such that the attention received increases with the fairness score $f_g$, yielding:
\[
    X_g^s = \sum_{i \in X^s} x_{u,i} \mathbb{I}(i \in \mathcal{I}_g),\quad \frac{\partial X_g^s}{\partial f_g} \ge 0.
\]    
\end{theorem}

\noindent%
A detailed proof is provided in Appendix~\ref{app:law_of_supply}. Intuitively, a higher price $f_g$ increases the likelihood of item $i$ capturing attention due to its attractiveness. Consequently, the attention $X_g^s$ obtained by provider $g$ is non-decreasing.

\subsection{Walrasian Equilibrium of two sides}

Next, we show that optimizing the fair re-ranking process is equivalent to seeking the Walrasian Equilibrium of the previously introduced demand and supply sides, where the fairness objective is viewed as a single total taxation.

First, fair re-ranking algorithms often seek to learn an optimal re-ranking score $f_g$, which corresponds to the price of attention defined earlier.
Next, we use the following theorem to prove that the fair re-ranking process is equivalent to the Walrasian Equilibrium of the demand and supply sides described above.

\begin{theorem}[Walrasian Equilibrium in Re-ranking]\label{theo:walrasian}
    Assuming $r'(v)$ is convex, we have:
   \begin{equation}
   \label{eq:opt_fair_score}
    \resizebox{0.92\linewidth}{!}{$
        \begin{aligned}
             \max_{x_{u,i}\in \mathcal{X}} \sum_u w_u - r'(v) \iff \max_{x_{u,i}\in \mathcal{X}} &\sum_u\sum_i (s_{u,i}-f^*_{g})x_{u,i},
        \end{aligned}
        $}
    \end{equation}
    where $\mathcal{X} = \{x_{u,i} | \sum_i x_{u,i}\leq K, ~x_{u,i}\in \{0,1\}\}$ and $w_u$ is the user utility defined in Section~\ref{sec:formulation}. 
    
    Moreover, assume that the price $f^*_{g}$ is such that 
   \[
   \sum_gX_g^d(f^*_{g}) = \sum_gX_g^s(f^*_{g}) = |\mathcal{U}|K.
   \]
   Then the allocation for demand and supply $\{X_g^d, X_g^s\}$ is Pareto optimal. That is, there does not exist another re-ranking allocation $\widetilde{X}^d, \widetilde{X}^s$ that satisfies:
   $\forall u, w_u (X^d) \leq w_u (\widetilde{X}^d)
   $,
    and 
    $\exists u_k, w_{u_k} (X^d) < w_{u_k} (\widetilde{X}^d).
    $
\end{theorem}

\noindent%
A detailed proof can be found in Appendix~\ref{app:walrasian}. Intuitively, if the fair re-ranking score $f^*_g$ (serving as the price) deviates from the optimum, market forces will correct it. The price will either rise in response to excess demand or fall in response to excess supply, driving the system toward improved accuracy and fairness.

Next, we use the following theorem to prove that the fairness learning objective is essentially equivalent to the single taxation mechanism. This result constitutes an equivalent formulation of the Second Fundamental Theorem of Welfare Economics~\cite{ng1983welfare}:

\begin{theorem}[Economic explaination for fairness]\label{theo:taxation_with_fairness}
   Assuming $r'(v)$ is convex, the fairness objective is equivalent to implementing a single taxation mechanism $T_g$. In this context, $T_g > 0$ represents a tax levied on ``rich'' providers, while $T_g < 0$ means a subsidy for ``poor'' providers. Eventually, it will reach the Pareto optimal point.
\end{theorem}

\noindent%
A detailed proof is provided in Appendix~\ref{app:taxation_with_fairness}. This theorem demonstrates that the concept of economic fairness aligns closely with the objectives of fair re-ranking. This theorem implies that a fundamental requirement for the fairness function is to redistribute utility, effectively taxing high-resource entities to compensate those in the lower-resource brackets.
This theoretical connection motivates our analysis of existing fair re-ranking algorithms.

\subsection{Regarding attention market as manifold}
In this section, we will use the manifold concept to describe the process for optimizing re-ranking.
In advanced economic theory, rather than viewing a market merely as a system of linear constraints or a static vector in $\mathbb{R}^{|\mathcal{G}|}$, we can model the entire set of feasible market states as a smooth manifold~\cite{balasko2009equilibrium}. 

Let the state of the economy be defined by a vector $x$ (representing allocations, prices, or attention distributions). The structural constraints of the economy define the geometry of $\mathcal{M}$:
\[
\mathcal{M} = \{ f \in \mathbb{R}^{|\mathcal{G}|} \mid \sum_g X_g^d(f) = \sum_g X_g^s(f) \}.
\]
Under this formulation, the ``Market'' is not just a point, but a geometric landscape (re-ranking scores and fairness objective) for finding the optimal price. Conducting gradient descent on such a manifold is subject to constraints from two sources: the supply side, governed by fairness parameters $\alpha$ and $\beta$, and the demand side, determined by the ranking score attributes $s_{u,i}$. Both factors influence the trajectory of the gradient step. 

However, existing online fair re-ranking methods often overlook these constraints within the context of Demand-Supply Equilibrium, leading to the unstable performance illustrated in Figure~\ref{fig:intro}. 

%In the following section, we analyze this interaction and propose a robust re-ranking model designed to address these limitations.

% And it can be aggregated to the general fairness form~\cite{ElasticRank}:
% \begin{equation}
%     \sum_g r(v_g) = r'(v) = \text{sign}(\alpha\beta-1)(\sum_g v_g^{\beta})^{\alpha}.
% \end{equation}    

%implying that the greater a provider’s gain $x$ from user attention, the higher the cost (i.e., tax) imposed on them.

\section{Method: ManifoldRank}

Building on the previously discussed theory, this section focuses on developing a method named ManifoldRank. 
The overall idea can be seen in the Algorithm~\ref{alg:ManifoldRank}.

\begin{algorithm}[t]
    \caption{Algorithm of ManifoldRank}
	\label{alg:ManifoldRank}
	\begin{algorithmic}[1]
	\REQUIRE User set $\mathcal{U}$, item set $\mathcal{I}$, group set $\mathcal{G}$, ranking size $K$, taxation parameters $\alpha, \beta$, demand parameters $a_e, a_s>0$, user-item ranking score $s_{u,i}, \forall u\in\mathcal{U},\forall i\in\mathcal{I}$ 
        \STATE Set $\bm{v}_g = 1, \forall g\in\mathcal{G}$
        \FOR{$u\in\mathcal{U}$}
            %\STATE Get fairness objective $r'(v) = \sum_g r(v_g)$ in Eq.~(\ref{eq:taxation_cost})
            \STATE  $// ~~ \texttt{Gradient compuation}$
            \STATE Valid direction $c_g = \sign\left((\alpha-1)\beta v_g^{\beta}+(\beta-1)(\sum_i v_g^{\beta})\right)$
            \STATE Get normalized factor $G = \sum_j x_j^{\beta}$
            \STATE Supply gradient: $\eta_g = c r(v_g)\beta\left(\frac{1}{v_g}+(\alpha-1)\frac{v_g^{\beta-1}}{G}\right)$
            %\STATE  $// ~~ \texttt{Demand term for gradient}$
            %\STATE Let $\bm{s} = [s_{u,i}]_{i=1}^{|\mathcal{I}|}$
            \STATE Demand gradient $\zeta = a_e\text{Entropy}(s) - a_s\text{Skew}(s)$
            \STATE  $// ~~ \texttt{Online re-ranking}$
            \STATE $L_K(u) = \argmax_{S\in\mathcal{I}^k} \sum_{i\in S} \left(s_{u,i}-\eta_g * \zeta\right), i\in \mathcal{I}_g$
            \STATE $\bm{v}_g = \bm{v}_g + \sum_{i\in L_K(u)} s_{u,i}\mathbb{I}(i\in \mathcal{I}_g), \forall g\in\mathcal{G}$
        \ENDFOR
	\end{algorithmic}
    %\vspace{-0.2cm}
\end{algorithm}

\subsection{Online re-ranking}

During online re-ranking, whenever a user $u$ enters the system, ManifoldRank will compute the fairness score:
\begin{equation}\label{eq:online_reranking}
    L_K(u) = \argmax_{S\in\mathcal{I}^k} \sum_{i\in S} \left(s_{u,i} - \eta_g(v_g) * \zeta(v_g) \right), 
\end{equation}
where $i\in \mathcal{I}_g$ and $v_g$ is updated after returning a ranked list:
\[
    \bm{v}_g = \bm{v}_g + \sum_{i\in L_K(u)} s_{u,i}\mathbb{I}(i\in \mathcal{I}_g), \forall g\in\mathcal{G}.
\]
The fairness score is composed of $\eta_g(v_g)$ and $\zeta(v_g)$, which can be interpreted as the gradients on the supply side and demand side, respectively. They will be detailed in the next two subsections. 

\emph{A comment on efficiency.} ManifoldRank derives the gradient in closed-form, necessitating only fundamental arithmetic operations. This ensures high computational efficiency and makes the approach well-suited for real-time online re-ranking.

\subsection{Gradient on the supply side}

From the supply side formulation in Section~\ref{sec:supplyside}, we can observe that the price (i.e., fairness score) is related to the cost function $r(v_g)$. Overall, the gradient can be written as:
\begin{equation}\label{eq:sup_gradient}
    \eta_g = c_g * f_g,
\end{equation}
where $c_g$ is the gradient direction and $f_g$ is the gradient value.

\textbf{Gradient value.} First, according to the supply side formulation in Eq.~(\ref{eq:supply}), we can 
observe that the price of getting attention $f_g$ is positively related to the first order of the cost function $r(v_g)$:
\begin{equation}
    \begin{aligned}
           f_g &= \frac{\partial r(v_g)}{\partial x_{u,i}} = \frac{\partial r(v_g)}{\partial v_g}\frac{\partial v_g}{\partial x_{u,i}} \propto \frac{\partial r(v_g)}{\partial v_g} = r(v_g)\gamma_g,   
    \end{aligned}
\end{equation}
where $\frac{\partial v_g}{\partial x_{u,i}}>0$ means allocating more user attention to the provider consistently enhances their utility, and $\gamma_g$ can be regarded as the manifold degree of curvature:
\[
    \gamma_g = \beta\left[\frac{1}{v_g}+(\alpha-1)\frac{v_g^{\beta-1}}{G}\right], G = \sum_g v_g^{\beta},
\]
where $G$ can be viewed as the taxation normalization for all groups.

\textbf{Gradient direction.} 
Theorem~\ref{theo:walrasian} and \ref{theo:taxation_with_fairness} both require that the total cost function $r'(v)$ is convex. Therefore, we need to adjust the direction $c_g = \{-1,1\}$ to make the term $\frac{\partial^2 r'(v)}{\partial v_g^2}$ be positive:
\begin{equation}
    c_g = \sign(\frac{\partial^2 r'(v)}{\partial v_g^2}) = \sign\left((\alpha-1)\beta x_i^{\beta}+(\beta-1)(\sum_i x_i^{\beta})\right).
\end{equation}

\subsection{Gradient on the demand side}\label{sec:gradient_demand}

The demand side presents a greater level of complexity. As outlined in Section~\ref{sec:demand}, the impact on price is associated with the ranking score $s_{u,i}$, and the constraints form a simplex structure. Consequently, we opt to use the statistical information of $s_{u,i}$ for each user $u$: $\bm{s} = [s_{u,i}]_{i=1}^{|\mathcal{I}|}$ to empirically assess the effect:
\begin{equation}
    \zeta = a_e\text{Entropy}(s) - a_s\text{Skew}(s),
\end{equation}
where \text{Entropy}(s), \text{Skew}(s) are the the entropy and skewness of $\bm{s}$ and $a_e>0$ and $a_s>0$ are the hyper parameters. 

Intuitively, higher entropy indicates greater diversity in user preferences, meaning that implementing fairness taxation on such users will have a minimal impact on performance. Conversely, higher skewness suggests that users prefer a very limited range of items, and imposing a higher taxation in this case could significantly harm performance. 
The influences of other statistical information is discussed in Section~\ref{exp:demand_info}.

\emph{Discussion.} 
We note that the connection to manifold-based optimization is an empirical approximation rather than a strict theoretical equivalence. While the supply-side formulation is derived analytically, the demand-side term $\zeta$ is constructed using statistical properties as a practical proxy for the complex simplex-constrained structure of user preferences. This design trades theoretical completeness for tractability and efficiency in online settings. Although heuristic, our experiments across diverse datasets suggest that this approximation captures stable and transferable patterns in practice, while we leave more principled formulations as future work.

\section{Experiments}

We evaluate ManifoldRank using seven public ranking datasets.

\subsection{Experimental settings}\label{sec:exp_settings}

\textbf{Dataset.} Our experiments are based on seven large-scale, public datasets: \textbf{Steam}~\cite{SASRec}: a ranking dataset for games on the Steam platform; \textbf{Five Amazon subsets}~\citep{he2016ups}: a widely used recommendation dataset, and \textbf{AliEC}\footnote{\url{https://tianchi.aliyun.com/dataset/56}}: a large-scale e-commerce ad recommendation dataset. The detailed information after pre-processing is in Table~\ref{tab:dataset}.

% $\bullet$
% \textbf{Steam}~\cite{SASRec}: a ranking dataset for games on the Steam platform.
% We use the data for games played for more than 10 hours in our experiments. The publishers of games are considered item groups. It has 169,160 samples, which contain 4,444 users, 1,237 items, and 43 item groups.\footnote{\url{http://cseweb.ucsd.edu/~wckang/Steam_games.json.gz}.}

% $\bullet$
% \textbf{Five Amazon subsets}~\citep{he2016ups}: a widely used recommendation dataset that consists of a variety of product categories and includes interactions such as product ratings, and purchase histories. They all use brands as item groups. After pre-processing, they include: 
% \emph{Digital music domains}: 46,640 samples, with 873 users, 11,952 items, and 33 item groups; \emph{Fashion domains}: 7,985 samples, with 629 users, 3,915 items, and 11 item groups; \emph{Software domains}: 13,115 samples, with 583 users, 2,707 items, and 34 item groups; \emph{Industrial and scientific domains}: 16,585 samples, with 1,078 users, 1,756 items, and 29 item groups; \emph{Arts, crafts, and sewing domains}: 113,365 samples, with 3,690 users, 3,109 items, and 100 item groups.\footnote{\url{http://jmcauley.ucsd.edu/data/amazon/}.}

% $\bullet$
% \textbf{AliEC}: a large-scale e-commerce ad recommendation dataset. The categories of items are considered as item groups. After pre-processing it has 505,430 samples, with 21,242 users, 17,310 items, and 117 item groups.\footnote{\url{https://tianchi.aliyun.com/dataset/56}}

\begin{table}[t]
\centering
\small
\caption{Statistics of datasets.}
\label{tab:dataset}
\setlength{\tabcolsep}{2pt}
\begin{tabular}{lccccccc}
\toprule
 & Steam & Music & Fashion & Software & Industrial & Arts & AliEC \\
\midrule
\#Users   & 4,444 & 873 & 629 & 583 & 1,078 & 3,690 & 21,242 \\
\#Items   & 1,237 & 11,952 & 3,915 & 2,707 & 1,756 & 3,109 & 17,310 \\
\#Groups  & 43 & 33 & 11 & 34 & 29 & 100 & 117 \\
\#Samples & 169,160 & 46,640 & 7,985 & 13,115 & 16,585 & 113,365 & 505,430 \\
\bottomrule
\end{tabular}
\end{table}

%During the pre-processing step, users that have interactions with fewer than $L$ items are excluded from the entire dataset to mitigate the issue of extreme sparsity. We set $L=8$ for Amazon-Arts-Crafts-and-Sewing, $L=5$ for AliEC, Steam, and Amazon-Industrial-and-Scientific, and $L=3$ for the other three datasets. Similarly, for the item side, we set $L=8$, $L=5$, and $L=1$, respectively. 
Following~\citet{ElasticRank}, we consider groups with fewer than 10 items as a single group, which we name the ``infrequent group''.
Following~\cite{xu2023p, cpfair}, we sort all interactions by time and use the first 80\% of the interactions as data to train the base ranking model. The remaining 20\% of interactions are used as the test data for re-ranking tasks.

\begin{table*}[ht]
\setlength{\tabcolsep}{3.4pt}
\Small
\caption{Performance comparison of ManifoldRank and the baselines on seven datasets, where ``Music'', ``Fashion'', ``Software'', ``Industrial'', and ``Arts'' denote five Amazon subsets. We tuned various re-ranking models to achieve accuracy performances (NDCG) exceeding 99\%, and evaluated fairness across two cut-offs $K$ to assess three fairness metrics (EF@K, GINI@K, MMF@K). Bold number indicates that the improvements over the baselines are statistically significant (t-tests and $p$-value $< 0.05$).}
\label{tab:EXP:main}
\centering
\begin{tabular}{@{}ll cc cc cc cc cc cc cc @{}}
\toprule
\multirow{2}{*}{Metric} & \multicolumn{1}{l}{Datasets} & \multicolumn{2}{c}{Steam} & \multicolumn{2}{c}{Music} & \multicolumn{2}{c}{Fashion} & \multicolumn{2}{c}{Software} & \multicolumn{2}{c}{Industrial} & \multicolumn{2}{c}{Arts} & \multicolumn{2}{c}{AliEC} \\ \cmidrule(l){2-16}
 & \multicolumn{1}{l}{Ranking size} & \multicolumn{1}{c}{K=10} & \multicolumn{1}{c}{K=20} & \multicolumn{1}{c}{K=10} & \multicolumn{1}{c}{K=20} & \multicolumn{1}{c}{K=10} & \multicolumn{1}{c}{K=20} & \multicolumn{1}{c}{K=10} & \multicolumn{1}{c}{K=20} & \multicolumn{1}{c}{K=10} & \multicolumn{1}{c}{K=20} & \multicolumn{1}{c}{K=10} & \multicolumn{1}{c}{K=20} & \multicolumn{1}{c}{K=10} & \multicolumn{1}{c}{K=20} \\
 \midrule
\multirow{6}{*}{EF} & CPFair & -12.797 & \underline{-2.0865} & -15.884 & -4.1494 & -4.6663 & -6.2038 & -13.487 & -2.0288 & -2.1845 & -2.0778 & -0.8414 & -1.1086 & -21.618 & -14.428 \\
 & Min-regularizer & -20.326 & -4.9795 & -15.907 & -6.0151 & -6.8100 & -5.8191 & -1.9039 & -3.0153 & -2.7667 & -2.5984 & -1.6493 & -1.1742 & -21.623 & -16.551 \\
 & P-MMF & -26.329 & -2.1071 & \underline{-2.8046} & -3.0857 & -4.5436 & -5.7765 & -1.9605 & -1.8360 & -1.6128 & -1.7352 & -1.1208 & -1.6035 & -8.7536 & -9.2907 \\
 & FairSync & \underline{-3.5548} & -5.9235 & -2.8208 & \underline{-2.7674} & -4.6886 & \underline{-4.4956} & \underline{-1.0825} & \underline{-1.2412} & \underline{-1.2458} & \underline{-1.1884} & \underline{-0.6494} & -0.8975 & -3.8817 & -5.5189 \\
 & ElasticRank & -4.9473 & -2.2201 & -4.9056 & -7.6306 & \underline{-4.5156} & -5.9379 & -1.4596 & -1.3225 & -1.2652 & -1.4249 & -0.6862 & \underline{-0.7460} & \underline{-3.5795} & \underline{-3.5835} \\
 & \textbf{ManifoldRank (ours)} & \textbf{-3.3971} & \textbf{-1.8908} & \textbf{-2.3834} & \textbf{-2.7253} & \textbf{-4.2612} & \textbf{-4.3041} & \textbf{-1.0263} & \textbf{-0.9437} & \textbf{-1.0860} & \textbf{-1.0774} & \textbf{-0.6276} & \textbf{-0.7013} & \textbf{-3.2473} & \textbf{-3.2589}\\
 \midrule
\multirow{6}{*}{GINI} & CPFair & \underline{0.6612} & 0.6444 & 0.7547 & 0.8590 & 0.7731 & 0.8191 & 0.4999 & 0.4928 & 0.5165 & \underline{0.5393} &\underline{0.2412} & 0.3978 & 0.7735 & 0.7689 \\
 & Min-regularizer & 0.6763 & 0.6418 & 0.7655 & 0.8614 & 0.8020 & 0.8192 & 0.5638 & 0.4726 & 0.5831 & 0.6008 & 0.3028 & 0.3918 & 0.7903 & \underline{0.7673} \\
 & P-MMF & 0.6675 & \underline{0.6222} & \underline{0.7537} & 0.8471 & \underline{0.7650} & 0.8172 & \underline{0.4740} & \underline{0.4342} & 0.5571 & 0.5348 & 0.2438 & 0.4694 & 0.8147 & 0.8059 \\
 & FairSync & 0.6985 & 0.6291 & 0.7629 & \underline{0.8419} & 0.7697 & \underline{0.8091} & 0.4863 & 0.4805 & \underline{0.4863} & 0.5475 & 0.2947 & 0.3958 & \underline{0.7621} & 0.7675 \\
 & ElasticRank & 0.7742 & 0.6377 & 0.8736 & 0.9005 & 0.8199 & 0.8490 & 0.5556 & 0.5547 & 0.5770 & 0.6162 & 0.2429 & \underline{0.3868} & 0.7937 & 0.8354 \\
  & \textbf{ManifoldRank (ours)} & \textbf{0.6608} & \textbf{0.6155} & \textbf{0.7363} & \textbf{0.8346} &\textbf{0.7495} & \textbf{0.7914} & \textbf{0.4539} & \textbf{0.3919} & \textbf{0.4608} & \textbf{0.5074} & \textbf{0.2226} & \textbf{0.3571} & \textbf{0.7065} & \textbf{0.7647} \\
 \midrule
\multirow{6}{*}{MMF} & CPFair & 0.0430 & 0.0840 & 0.0745 & 0.0664 & 0.0527 & 0.0358 & 0.1706 & 0.1767 & 0.1700 & 0.1618 & 0.3377 & 0.3041 & 0.0490 & 0.0487 \\
 & Min-regularizer & 0.0259 & 0.0567 & 0.0945 & 0.0554 & 0.0415 & 0.0361 & 0.1374 & 0.1871 & 0.1364 & 0.1317 & 0.2982 & 0.2990 & 0.0552 & 0.0491 \\
 & P-MMF & 0.0530 & \underline{0.1053} & 0.0957 & \underline{0.0771} & \underline{0.0637} & 0.0476 & 0.2074 & \underline{0.2304} & 0.1607 & 0.1815 & 0.3425 & \underline{0.3218} & 0.0424 & 0.0433 \\
 & FairSync & \underline{0.0688} & 0.0794 & \underline{0.1167} & 0.0764 & 0.0609 & \underline{0.0532} & \underline{0.2101} & 0.2111 & \underline{0.2140} & \underline{0.1830} & 0.3238 & 0.2586 & \underline{0.0721} & 0.0590 \\
 & ElasticRank & 0.0452 & 0.0915 & 0.0465 & 0.0284 & 0.0427 & 0.0319 & 0.1879 & 0.1902 & 0.1721 & 0.1533 & \underline{0.3577} & 0.3124 & 0.0680 & \underline{0.0640} \\
  & \textbf{ManifoldRank (ours)} & \textbf{0.0754} & \textbf{0.1096} & \textbf{0.1255} & \textbf{0.0799} & \textbf{0.0706} & \textbf{0.0620} & \textbf{0.2227} & \textbf{0.2617} & \textbf{0.2190} & \textbf{0.2005} & \textbf{0.3621} & 0.3158 & \textbf{0.0774} & \textbf{0.0704} \\
 \bottomrule
\end{tabular}
\end{table*}

\textbf{Evaluation.} The performance of the models is evaluated on two aspects: re-ranking accuracy and fairness degree. 
For accuracy, following~\cite{ElasticRank, xu2023p}, we use NDCG@K:
\begin{align}
   \text{NDCG@K} & =\frac{1}{|\mathcal{U}|} \sum_{u\in\mathcal{U}}\frac{\sum_{i\in\mathbf{L}_K^F(u)}s_{u,i}/\log(\textrm{rank}^F_i+1)}{\sum_{i\in\mathbf{L}_K(u)}s_{u,i}/\log(\textrm{rank}_i+1)},
\end{align}
where $\mathbf{L}_K(u_t)$ is the original ranked list and $\mathbf{L}^F_K(u_t)$ is the fair-aware re-ranked list, and rank$_i$ and rank$_i^F$ are the ranking positions of the item $i$ in $\mathbf{L}_K(u_t)$ and $\mathbf{L}_K^F(u_t)$, respectively.  

%The improved re-ranking accuracy results in a higher NDCG@K value and a lower Loss@K value, indicating better re-ranking quality.

% NDCG@K is defined as the ratio between the sum of position-based user-item scores~\cite{wu2021tfrom} in the original ranking list $\mathbf{L}_K(u_t)$ and that in the re-ranked list $\mathbf{L}^F_K(u_t)$:

To evaluate fairness, we use three distinct metrics: Elastic Fairness (EF)~\cite{ElasticRank}, the Gini Index (GINI)~\cite{nips21welf}, and Max-Min Fairness (MMF)~\cite{xu2023p}. We adopt the MMF@K formulation from \cite{nips21welf}, which quantifies the aggregated ranking score of the bottom 20\% ``worst-off'' groups, thereby measuring the system's ability to protect the most disadvantaged entities. Higher values for EF@K and MMF@K, coupled with a lower GINI@K, indicate a more fair distribution of provider groups.

% \begin{align}
%     EF@K = \int_{1-M}^{1+M} \frac{\text{sign}(1-t)\left(\sum_{g=1}^{|\mathcal{G}|} \bar{\bm{v}}_g^{1-t}\right)^{\frac{1}{t}}}{2M|\mathcal{G}|} dt, 
% \end{align}
% where we set $M=50$ and $\bar{v}_g$ is the normalized utilities for group $g$ under ranking size $K$. 

\textbf{Baselines.} The following representative online re-ranking models were chosen as baselines: 
\emph{CPFair}~\cite{cpfair} formulates the re-ranking task as a knapsack problem and employs a greedy strategy;
\emph{Min-regularizer}~\cite{xu2023p} introduces a regularization term that penalizes exposure gaps between target providers and worst-off groups;
\emph{P-MMF}~\cite{xu2023p} uses mirror gradient descent specifically to improve the utility of the worst-off item groups;
\emph{ElasticRank}~\cite{ElasticRank}, which determines re-ranking priority by calculating a fairness score based on the exposure elasticity between two item groups; and
\emph{FairSync}~\cite{fairsync}, which implements an online gradient descent framework designed to guarantee minimum exposure thresholds for every item group.

\textbf{Implementation details.} Our experiments are implemented using Python 3.9. All experiments are conducted on a server with Ubuntu 18.04. All our experiments are conducted based on the \emph{FairDiverse} toolkit~\cite{xu2025fairdiverse}.
As for the hyperparameters in all models, the fairness parameter $\beta, \alpha$ is tuned among $[-3,3]$. The demand parameter $a_e, a_s$ is tuned among $[0.1, 1]$.

\subsection{Experimental results}
We report on the performance of ManifoldRank and the baselines. Our main experiments are conducted based on the well-performing ranking base model SASRec~\cite{SASRec}.

\subsubsection{Fairness performance comparison.}
To ensure a rigorous evaluation, we follow the methodology established by ElasticRank. We conduct experiments to assess the trade-off between accuracy and fairness for ManifoldRank and our selected baselines. Specifically, we constrain the accuracy (measured by NDCG@K) to remain above 99\% of the base model's performance. This approach is motivated by the fact that an accuracy loss of less than 1\% is typically considered negligible in production environments. 

Comprehensive results are shown in Table~\ref{tab:EXP:main} details the experimental outcomes for ManifoldRank and the baseline methods. We use datasets with different data distributions to ensure the robustness of our model.
From Table~\ref{tab:EXP:main}, we can observe that ManifoldRank outperforms almost all baselines in terms of all fairness metrics on all datasets (except for 1 out of 42 settings), including EF@K, GINI@K, and MMF@K. The experimental results demonstrate that ManifoldRank achieves superior fairness performance at the same accuracy level, highlighting our effectiveness and robustness.

\begin{figure}[t]
    \centering
    \includegraphics[width=0.98\linewidth]{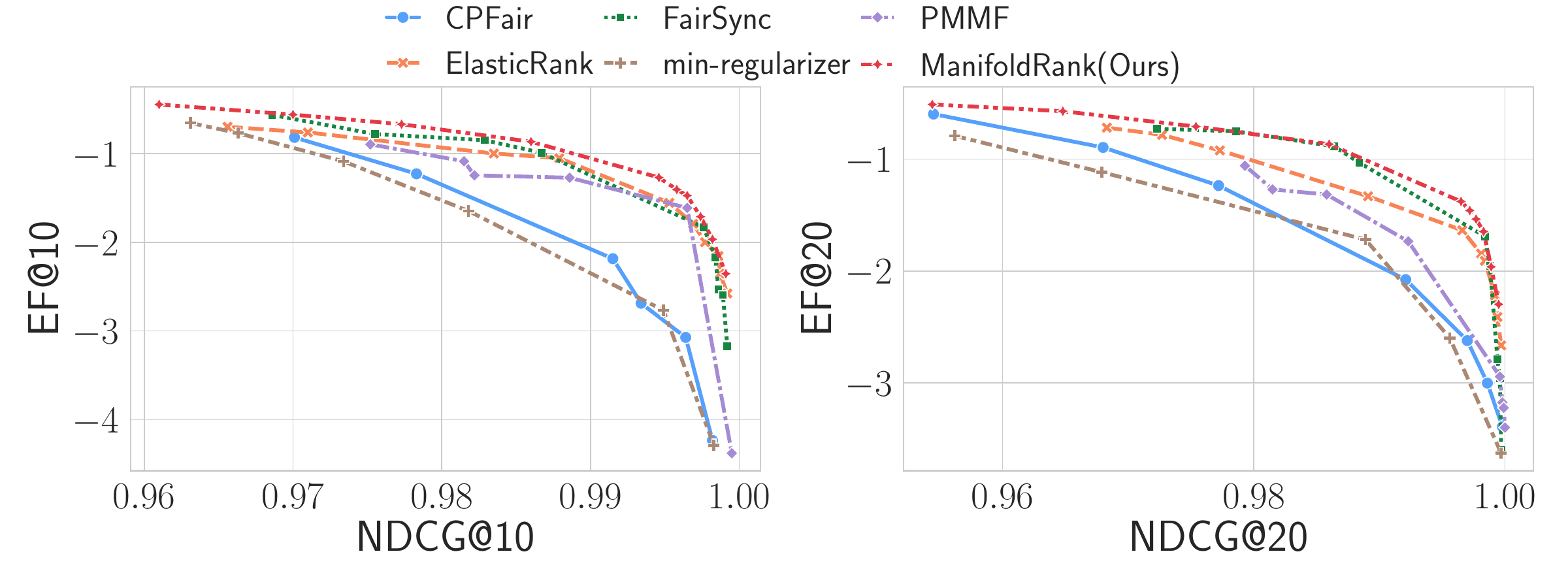}
    \caption{Pareto frontier with different size $K$ under Amazon industrial and scientific domains on SASRec. }
    \label{fig:Pareto}  
\end{figure}

%From Table~\ref{tab:EXP:main}, we first observe that \textbf{FairRec}~\cite{fairrec}, \textbf{FairRec+}~\cite{fairrecplus}, and \textbf{TaxRank}~\cite{TaxRank}, as item-level fairness methods, struggle to effectively balance accuracy and group-level fairness, often resulting in either insufficient accuracy or inadequate fairness. For the remaining baselines, the experimental results clearly demonstrate that ElasticRank achieves superior fairness performance at the same accuracy level, highlighting the effectiveness of our model.

%Next, we test the performance of ElasticRank and the baselines for different accuracy-fairness trade-off degrees and various fairness metrics by conducting experiments on the Steam dataset. Similar trends can be observed in the other two datasets.

\subsubsection{Performance on Pareto frontiers.} 
Figure~\ref{fig:Pareto} illustrates the Pareto frontiers~\cite{lotov2008visualizing} between the accuracy metric NDCG@K and the fairness metric EF@K at $K=10$ and $K=20$. Due to space limitations, we report results on the Amazon Industrial and Scientific domains; similar patterns are consistently observed across other datasets and base models.
The Pareto frontiers are obtained by systematically varying model parameters and selecting configurations that jointly optimize NDCG@K and EF@K, thereby characterizing the trade-off between ranking accuracy and item fairness. As expected, a clear trade-off emerges: improving fairness often comes at the cost of accuracy, and vice versa.
Compared with baseline methods, the proposed ManifoldRank consistently achieves Pareto-superior performance, as its curves occupy the upper-right region of the frontier. This dominance indicates that, for a given NDCG@K level, ManifoldRank attains higher EF@K, and for a given EF@K level, it delivers better NDCG@K. These results demonstrate the substantial advantage of ManifoldRank over existing baselines.

\subsection{Experimental analysis}

In this section, we report on an experimental analysis on Amazon Industrial and Scientific domains; similar patterns are consistently observed across other datasets and base models.

\subsubsection{Ablation study on different base models}\label{sec:ab4BM}
To verify the generalizability of ManifoldRank, we integrate it with three distinct base ranking models to obtain the preliminary preference scores $s_{u,i}$. These include: (i)~BPR~\cite{BPR}, a pairwise collaborative filtering model; (ii)~GRU4Rec~\cite{gru4rec}, a recurrent neural network designed for session-based patterns; and (iii)~SASRec~\cite{SASRec}, a transformer-based model that uses attention mechanisms for sequential modeling.

Table~\ref{tab:base_model} presents a comparative analysis of ManifoldRank against the two strongest baselines, FairSync and ElasticRank, across various base ranking architectures. Following the experimental configuration established in Table~\ref{tab:base_model}, we observe that ManifoldRank consistently yields superior performance. These results validate the robustness of our framework, demonstrating that its effectiveness is independent of the underlying base model. This performance gain is largely attributed to our gradient formulation, which explicitly integrates the base model's preference scores into the demand-side gradient calculation.

\begin{table}[ht]
\centering
\caption{Fairness performances for different base model BPR, GRU4Rec, SASRec. Also, we tuned various re-ranking models to achieve accuracy performances (NDCG) exceeding 99\%, and evaluated fairness across two cut-offs $K$ to assess three fairness metrics (EF@K).}
\label{tab:base_model}

\small
\setlength{\tabcolsep}{3pt}
\begin{tabular}{@{}l cc cc cc@{}}
\toprule
\multicolumn{1}{@{}l}{\multirow{2}{*}{Model}} & \multicolumn{2}{c}{BPR} & \multicolumn{2}{c}{GRU4Rec} & \multicolumn{2}{c}{SASRec} \\ 
\cmidrule(r){2-3} 
\cmidrule(r){4-5}
\cmidrule{6-7}
\multicolumn{1}{c}{} & \multicolumn{1}{c}{EF@10} & \multicolumn{1}{c}{EF@20} & \multicolumn{1}{c}{EF@10} & \multicolumn{1}{c}{EF@20} & \multicolumn{1}{c}{EF@10} & \multicolumn{1}{c}{EF@20} \\ \hline
FairSync  & -0.1861 & -0.5712 & -0.9317 & -1.0366 & -1.0610 & -1.0326\\
ElasticRank & -0.0136 & -0.1442 & -1.1096 & -1.1680 & -1.0514 & -1.3328\\
\textbf{ManifoldRank} & \textbf{-0.0120} & \textbf{-0.0412} & \textbf{-0.7876} & \textbf{-0.9635} & \textbf{-0.9124} & \textbf{-0.9088}\\ 
\bottomrule
\end{tabular}
\end{table}

\subsubsection{Ablation study on fairness parameters}

We conduct an ablation analysis to investigate the influence of the fairness parameters $\alpha$ and $\beta$. We evaluate in EF@10 and NDCG@K by varying each parameter individually while fixing its counterpart; see Figure~\ref{fig:ablation}.

\begin{figure}[h]
    \centering
    \includegraphics[width=0.98\linewidth]{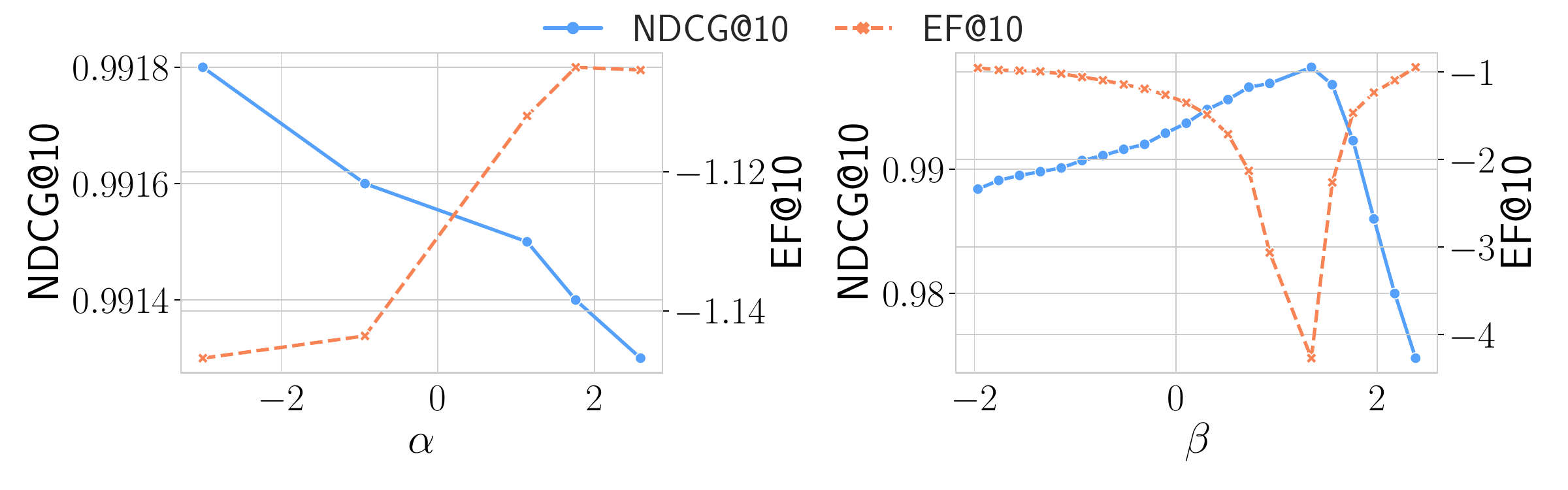}
    \caption{Ablation study on NDCG@K and EF@K performance \wrt fairness parameter $\alpha, \beta$. }
    \label{fig:ablation}  
\end{figure}

As illustrated in Figure~\ref{fig:ablation} (left), an increase in $\alpha$ leads to a corresponding improvement in the fairness metric EF@10, albeit at the expense of a decrease in NDCG@10. This observation aligns with the theoretical analysis presented in Section~\ref{sec:theory}, as $\alpha$ governs the global taxation degree; a larger $\alpha$ effectively amplifies the gradient associated with the fairness objective.

Second, Figure~\ref{fig:ablation} (right) illustrates a non-monotonic trend for the local fairness parameter $\beta$. As $\beta$ increases within the range $[-2, 1]$, we observe a decline in EF@10 accompanied by an improvement in NDCG@10. However, as $\beta$ continues to rise from $[1, 3]$, this trend reverses: EF@10 increases while NDCG@10 decreases. This behavior aligns with the theoretical framework in Section~\ref{sec:theory}; specifically, when $\beta < 1$, decreasing its value prioritizes the utility of the ``worst-off'' providers. Conversely, when $\beta > 1$, larger values impose a heavier penalty on ``rich'' providers, thereby progressively enhancing global fairness.

Finally, we observe that the model exhibits lower sensitivity to $\alpha$, suggesting that this parameter requires tuning over a broader range to induce significant performance shifts. In contrast, the local fairness parameter $\beta$ demonstrates a more pronounced impact on fairness metrics; this heightened sensitivity stems from its capacity to exert a larger influence on the gradient magnitude.

\subsubsection{Performance at each online step}

In this section, we provide a comparative analysis of the real-time accuracy and fairness trajectories between our proposed ManifoldRank and the strongest baseline, FairSync. The performance evolution is depicted in Figure~\ref{fig:online_step}.

\begin{figure}[h]
    \centering
    \includegraphics[width=0.98\linewidth]{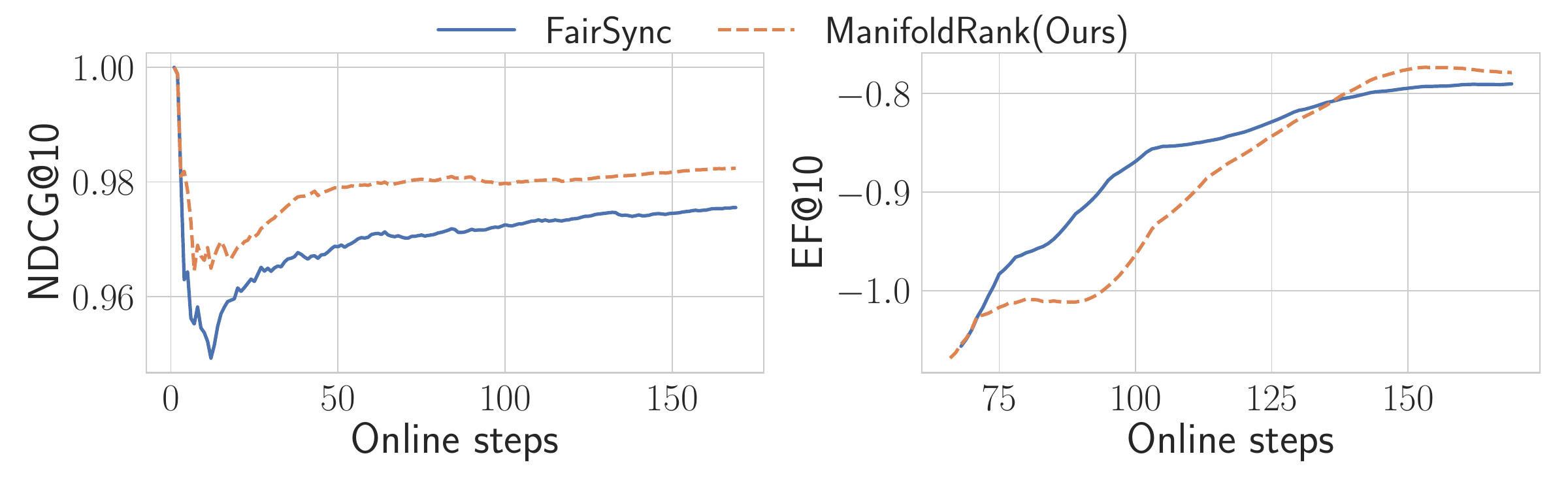}
    \caption{NDCG@K and EF@K performance \wrt online step. }
    \label{fig:online_step}  
\end{figure}

\noindent%
As illustrated in Figure~\ref{fig:online_step}, both ManifoldRank and FairSync initially prioritize the fairness objective, which leads to a transient decline in NDCG. However, as the online procedure progresses, both frameworks successfully recalibrate, demonstrating a simultaneous improvement in fairness and accuracy. 

Compared to FairSync, ManifoldRank consistently maintains superior accuracy throughout the online process (left panel), reflecting high utility on the demand side. Regarding fairness (right panel), while ManifoldRank initially exhibits a more gradual improvement trend within the first 140 steps, it eventually surpasses the baseline to achieve superior long-term fairness. These results underscore ManifoldRank's suitability for real-world applications; in large-scale deployment scenarios involving massive user bases and extended operational timelines.

%This convergence suggests that the models effectively learn to balance competing objectives, ultimately achieving a superior performance equilibrium.

\subsubsection{Influence on information of demand side}\label{exp:demand_info}

To investigate the influences of different statistical properties of the demand side vector $\bm{s}$, we tune ManifoldRank without the demand-side gradient $\zeta$ using different ranking scores, and frame a regression task in which the fairness metric EF@K is treated as the target variable. Then we denote by $f_i(\mathbf{s})$ the $i$-th statistical characteristic of $\mathbf{s}$, which is used as a regression feature and includes measures such as entropy, skewness, kurtosis, and related statistics:
$
    EF@K = a_if_i(s) + \epsilon,
$
where $\epsilon$ is the noise variable of regression.

Figure~\ref{fig:regression_res_plot} presents the estimated coefficients $a_i$ together with their confidence intervals, where ``mean'' and ``var'' in the legend denote the mean and variance across different users. From the figure, we observe that the mean values of Skewness and Entropy have a pronounced influence on the online re-ranking process, whereas the other factors remain close to zero, indicating limited impact. Therefore, in Section~\ref{sec:gradient_demand}, we select Skewness and Entropy as proxy properties of the manifold.

Meanwhile, the coefficient of Skewness is negative, whereas that of Entropy is positive. Higher entropy reflects greater diversity in user preferences, so enforcing fairness has little impact on performance. In contrast, higher skewness indicates concentrated preferences, and stronger taxation in this case may substantially degrade performance.

%We then use $f_i(\bm{s})$ to perform an ordered regression analysis on $y$ to identify the statistical measures that have the greatest impact on the re-ranking process:

% \[
% \max P(y \leq j \mid f_i(\bm{s})) = \frac{1}{1 + \exp\left(-(b_j - \sum_i f_i(\bm{s})a_i)\right)}, j=1,2, \cdots.
% \]

\begin{figure}
    \centering
    \includegraphics[width=\linewidth]{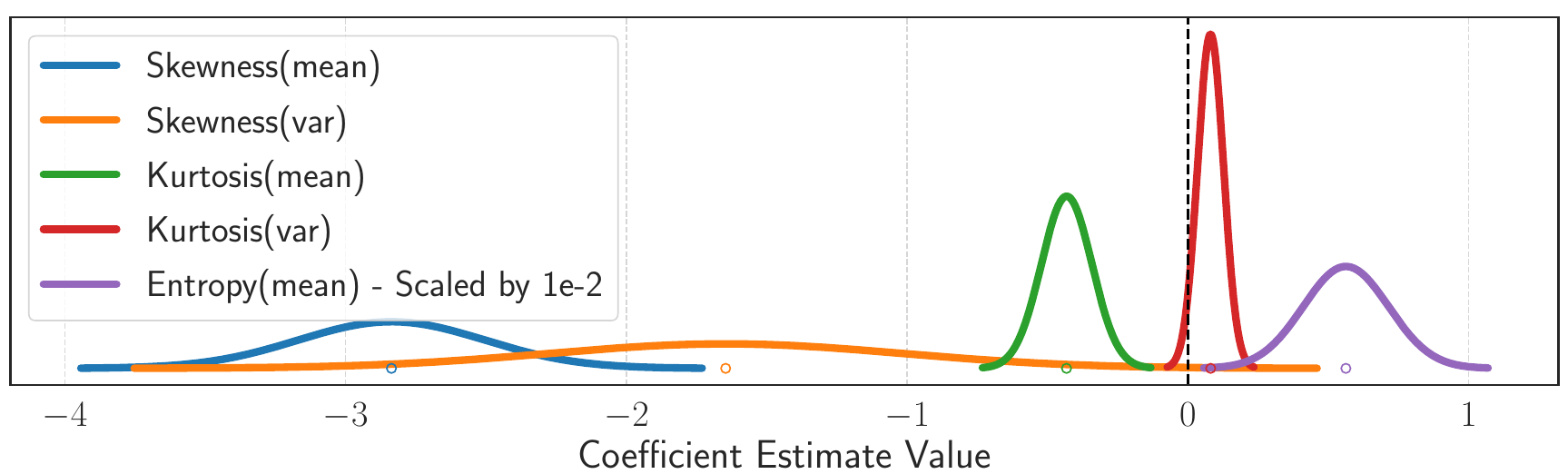}
    \caption{Influence of the statistical properties of the demand-Side vector $\bm{s}$, where ``mean'' and ``var'' in the legend denote the mean and variance across different users.}
    \label{fig:regression_res_plot}
    \vspace{-0.2cm}
\end{figure}
\section{Conclusion}

In this paper, we first formulate fair re-ranking as a Walrasian Equilibrium within an attention economy. We demonstrate that varying market contexts have a significant influence on the generalizability of online re-ranking gradients on the manifold. Building on these theoretical insights, we propose Manifold. This advanced online fair re-ranking algorithm dynamically adjusts gradients by accounting for both the supply and demand sides of the market. Extensive empirical evaluations show the effectiveness of ManifoldRank. As to limitations, our theory is still only based on a fixed ranking size under static settings. In future work, we will extend the theory into a dynamic and long-term re-ranking conditions.
%\mdr{What are the implications?}

%We are aware of the previous work
%While prior research has explored fair re-ranking through an economic lens \cite{TaxRank, ElasticRank}, this paper provides a more rigorous theoretical formalization of the underlying market dynamics. We demonstrate that fairness interventions are far more than isolated technical adjustments; rather, they represent a fundamental recalibration of the market’s distributive mechanics. 

%By establishing a formal framework, we analyze how supply-side perturbations—introduced by fairness constraints—propagate through the system to reshape the overall economic welfare of the ecosystem.

%As to limitations, our theory is still only based on a fixed ranking size under static settings. In future work, we will extend the theory into a dynamic and long-term re-ranking conditions.

%understand accuracy-fairness trade-offs in re-ranking by framing them as a commodity taxation transfer problem. By leveraging elasticity theory from economics, we reveal that these trade-offs are determined by the elasticity between inter-groups. Inspired by elasticity theory, we introduce the EF-Curve, a evaluation framework for fair re-ranking, alongside ElasticRank, an algorithm that consistently outperforms state-of-the-art baselines in both efficiency and effectiveness.

\begin{acks}
This research was (partially) supported by the National Natural Science Foundation of China - 62472426, the Dutch Research Council (NWO), under project numbers 024.004.022, NWA.1389.20.\-183, and KICH3.LTP.20.006, and the European Union under grant agreement No. 101201510 (UNITE).
Views and opinions expressed are those of the author(s) only and do not necessarily reflect those of their respective employers, funders, and/or granting authorities.
\end{acks}

\appendix
\section*{Appendix}

\section{Proof of Theorem~\ref{theo:law_of_demand}}\label{app:law_of_demand}

\begin{proof}
Based on the condition, we have:
$
    \frac{\partial X_g^d}{\partial f_{g}} = \mathbb{I}(i \in \mathcal{I}_g) \cdot \frac{\partial x_{u,i}^*}{\partial f_{g}},
$
where $x_{u,i}^*$ is the optimal attention to item $i$ by user $u$.

Now, the $x_{u,i}^*$ is decided by the $\sum_i (s_{u,i} - f_{g}) x_{u,i}$ subject to the budget and box constraints. This is a linear program (LP) problem and by the monotonicity of optimal solutions in LP, we have:
$
    \frac{\partial x_{u,i}^*}{\partial f_{g}} = -\lambda_u \leq 0,
$
where $\lambda_u \geq 0$ is the dual multiplier for the budget constraint $\sum_i x_{u,i} \leq K$. Thus, $\frac{\partial x_{u,i}^*}{\partial f_{g}} \le 0$, implying
$
    \frac{\partial X_g^d}{\partial f_{g}} \le 0.
$
\end{proof}

\section{Proof of Theorem~\ref{theo:taxation}}\label{app:taxation}

\begin{proof}
    First, the total supply cost can be written as:
    $r'(v) = \sum_g r(v_g) =  \sign(\alpha\beta)(\sum_i v_g^{\beta})^{\alpha}$.
    Let $G = \sum_g v_g^{\beta}$ ($v_g$), then we have:
    \begin{align*}
        \frac{\partial r'(v)}{\partial v_g} = \text{sign}(\alpha\beta)\alpha\beta G^{\alpha-1}v_g^{\beta-1}.
    \end{align*}
    Since $v_g \ge 0$, we have: 
    $
        G = \sum_g v_g^{\beta} \ge 0, v_g^{\beta-1} \ge 0,
    $
    then $\frac{\partial r'(v)}{\partial v_g} \ge 0$.
\end{proof}

\section{Proof of Theorem~\ref{theo:law_of_supply}}\label{app:law_of_supply}

\begin{proof}
    Similar to the proof of Theorem~\ref{theo:law_of_demand}, we have
    $    
    \frac{\partial X_g^s}{\partial f_{g}} = \mathbb{I}(i \in \mathcal{I}_g) \cdot \frac{\partial x_{u,i}^*}{\partial f_{g}},
    $
    And the $x_{u,i}^*$ is decided by the $\sum_i (f_{g} x_{u,i}) - v(v_g)$ subject to the constraints. By the monotonicity of LP, we have:
    $
    \frac{\partial x_{u,i}^*}{\partial f_{g}} = \lambda_u \ge 0,
    $
   where $\lambda_u \geq 0$ is the dual multiplier, implying
    $
        \frac{\partial X_g^s}{\partial f_{g}} \ge 0.
    $
\end{proof}

\section{Proof of Theorem~\ref{theo:walrasian}}
\label{app:walrasian}

\begin{proof}

\emph{Step 1:}
First, we will prove that $f_g^*$ is the price to make demand and supply $\{X_g^d, X_g^s\}$ Pareto optimal. 
Analogous to the proof strategy of the First Fundamental Theorem of Welfare Economics~\cite{ng1983welfare}, assuming there exists $\widetilde{X}^d, \widetilde{X}^s$ for the user that satisfies:
\[
   w_u (X^d) \leq w_u (\widetilde{X}^d),\forall u, \quad \exists u_k, w_{u_k} (X^d) < w_{u_k} (\widetilde{X}^d).
\]
Then we have the price $f_g^*$ for the supply side aims to maximize their attention gained:
$
    f^*_g \widetilde{X}^s_g \leq f^*_g X^s_g, \quad \sum_g f^*_g \widetilde{X}^s_g \leq \sum_g f^*_g X^s_g.
$
According to the market clear condition $\widetilde{X}^d_g=\widetilde{X}^s_g$, we have:
\begin{equation}\label{eq:supply_market_clear}
    \sum_g f^*_g \widetilde{X}^d_g \leq \sum_g f^*_g X^d_g.
\end{equation}
While for the demand side, $\exists u_k, w_{u_k} (X^d) < w_{u_k} (\widetilde{X}^d)$ must satisfy that $f^*_g X^d_g > f^*_g\widetilde{X}^d_g$, otherwise it does not satisfy the maximize user's utility condition:
\[
     w_{u_k} (X^d) = \sum_{i\in\mathcal{I}_g}(s_{u,i}-f_g^*)x_{u,i} < \sum_{i\in\mathcal{I}_g}(s_{u,i}-f_g^*)\widetilde{x}_{u,i} = w_{u_k} (\widetilde{X}^d).
\]
Such a condition will lead to
$
    \sum_g f^*_g \widetilde{X}^d_g > \sum_g f^*_g X^d_g,
$
which contradicts Eq.~(\ref{eq:supply_market_clear}).

\emph{Step 2:} Then we will show that the fair re-ranking process can be regarded as the Walrasian Equilibrium seeking process. 
Let $|\mathcal{U}|K = T$, the total social welfare function can be viewed as:
\begin{align*}
    W = \max_{x_{u,i}\in \mathcal{X}} &\sum_u w_u - \sum_g r(v_g), ~ \sum_g X^d_g = \sum_g X^d_s = T~ (\textit{Equilibrium}).
\end{align*}
Then we move the constraints to the objective using a vector of Lagrange multipliers $f_g$ (\ie the price for fair re-ranking):
\begin{equation}
         W = \max_{x_{u,i}\in \mathcal{X}}\left(\min_{f_g^d}d(f_g^d) + \min_{f_g^s}s(f_g^s)\right)
\end{equation}         
\begin{equation}
    \begin{aligned}
    \textrm{s.t. }  d(f_g^d) &= \sum_u w_u + \sum_g f_g^d(\sum_g X^d_g-T), &(\textit{Demand function}) \\
      d(f_g^s) &= \sum_g f_g^s(\sum_g X^d_s-T) - \sum_g r(v_g), &(\textit{Supply function}) \\
      f_g^s &= f_g^d.  &(\textit{Market condition})
    \end{aligned}
    \nonumber
\end{equation}
Such a demand function will lead to the conclusion.
\end{proof}

\section{Proof of Theorem~\ref{theo:taxation_with_fairness}}
\label{app:taxation_with_fairness}

\begin{proof}
    We define the aggregate production set $S$ as:
    $ S = \left\{ \sum_{g} X^s_g \right\} $. Define the aggregate ``better-than'' set $D$:
    $$ D = \left\{ \sum_{g}^I X_g^d \mid w_u(X_g^d) \geq w_u(X^*), \exists g, w_u(X_g^d) > w_u(X^*) \right\}.$$
    Since $X^*$ is Pareto optimal, the intersection of the interior of $D$ and the set $S$ must be empty:
    $ S \cap D = \emptyset. $
    By the Separating Hyperplane Theorem~\cite{mclennan2002ordinal}, there exists a non-zero price vector $f \in \mathbb{R}^{|\mathcal{G}|}$ and a scalar $k$:
    \begin{equation}
        f \cdot d \geq k \geq f \cdot s \quad \forall d \in D, \forall s \in D.
    \end{equation}
    Then this single price vector $f$ supports individual optimization for getting the optimal $X^s_g$.  And the parameters $\beta, \alpha$ will change the objective, eventually reach the optimal price $f$. Due to the constraints of the re-ranking setting, the total produced attention is bounded by the available slots $\sum_g X^s_g \leq |\mathcal{U}|K$. 
    
    Adjusting the price $f$ shifts the optimal attention allocation, which is effectively implemented through a single taxation mechanism $T_g$. In this context, $T_g > 0$ represents a tax levied on ``rich'' providers, while $T_g < 0$ means a subsidy for ``poor'' providers:
    \[
        \max\{0,T_g\} \leq X^s_g \leq X_g^0-\max\{0,T_g\}, \quad \sum_g T_g = 0,
    \]
    where $X_g^0$ denotes the attention produced without the fairness.
\end{proof}

\clearpage
\bibliographystyle{ACM-Reference-Format}
\balance
\bibliography{references}

\end{document}